\RequirePackage{fix-cm}
\documentclass[smallextended]{svjour3_custom}       
\smartqed  
\usepackage{graphicx}
\usepackage{natbib}
\usepackage{bbm}
\usepackage{amsmath}
\usepackage{ntheorem}
\newcommand{\E}{{\rm E}}

\renewcommand{\Pr}{{\rm Pr}}

\newcommand{\Var}{{\rm Var}}
\newcommand{\Cov}{{\rm Cov}}

\newtheorem*{lemma*}{Lemma}
\journalname{ArXiv}
\begin{document}

\title{Estimation of the regression slope by means of Gini's cograduation index}


\titlerunning{Estimation of the regression slope by means of Gini's cograduation index}        

\author{D. Michele Cifarelli}


\institute{D. Michele Cifarelli \at
              Istituto di Scienze Statistiche e Demografiche \\
              Universit\`a di Pavia (Italy)\\
\emph{Present address:} \\
Department of Decision Sciences,
Bocconi University,   
via Roentgen 1,
20135 Milano (Italy)
\email{michele.cifarelli@unibocconi.it}          
}

\date{{\bf Translation from  Italian of the paper:} Cifarelli, D.M. (1978). ``La stima del coefficiente di regressione mediante l'indice di cograduazione di Gini", {\em Rivista di matematica per le scienze economiche e sociali} (now: {\em Decisions in economics and finance}),  1, 7--38.}


\maketitle

\begin{abstract}
The simple linear model 
$$Y_i = \alpha + \beta \, x_i + \epsilon_i \qquad i=1,2, \ldots,N \geq 2$$
is considered, where the $x_i$'s are given constants and $\epsilon_1, \epsilon_2 , \ldots, \epsilon_N$ are iid with continuous distribution function $F$. An estimator of $\beta$ is proposed, based on the stochastic process in (\ref{due}) and defined as $\tilde{\beta} = \frac 12 \, \left\{ \sup (b: G(\underline y;b) >0) + \right. $ 
$ \left.  \inf (b: G(\underline y;b) <0) \right\}.$ The properties of $\tilde{\beta}$ and of the related confidence interval are studied. Some comparisons are given, in terms of asymptotic relative efficiency, with other estimators of $\beta$ including that obtained with the method of least squares.    
\end{abstract}

\section{Introduction and summary}

Consider the simple linear model
$$
Y(x_i) = Y_i = 
\alpha + \beta \, x_i + \epsilon_i \qquad i=1,2, \ldots,N,$$
where
\begin{enumerate}
\item[a)] $x_1, x_2, \ldots , x_N$ are known constants, supposed to be all distinct and increasingly ordered
\item [b)] $\epsilon_1, \epsilon_2, \ldots, \epsilon_N$ are mutually independent random variables with the same distribution function $F$
\item [c)] $\alpha$ and $\beta$ are unknown parameters.
\end{enumerate}

The usual estimators of $\alpha$ and $\beta$ are those derived from the least squares method. As known, if the $\epsilon_i$'s have finite variance, such estimators possess some good properties. More specifically, they are unbiased and have minimum variance in the class of linear estimators (BLUE). When, in addition, the $\epsilon_i$'s are assumed to be normal, the above estimators coincide with the ones obtained by the maximum likelihood method and, besides being unbiased, they have minimum variance in the class of all unbiased estimators (MVUE) and they are normally distributed. 

Consider then the least squares estimator of $\beta$:
\begin{equation} \label{uno}
\hat{\beta} = \frac{\displaystyle{\sum_{1\leq i \leq N}} (x_i - \bar x) (Y_i - \bar Y)}{\displaystyle{\sum_{1\leq i \leq N}} (x_i - \bar x)^2} \qquad \bar Y = \frac 1N \sum_{1\leq i \leq N} Y_i, \quad \bar x = \frac 1N \sum_{1\leq i \leq N} x_i.
\end{equation} 
As the corresponding estimate of $\beta$ strongly depends on the observed values $y_1, y_2, \ldots , y_N,$ the occurrence of outliers, that is of observations deviating from the main core of data, will likely influence such a procedure. This chance will often arise when the distribution of the disturbances $\epsilon_i$ has heavy tails, like  in the case of the Cauchy, the double-exponential and other distributions. It is  quite a serious drawback of the estimator $\hat{\beta}$ and attempts are occasionally made to remedy it by unconventionally deleting the most extreme observations. 

Another completely different problem of least squares concerns the interval estimation of $\beta.$ The possibility of producing a confidence interval for $\beta,$ or equivalently of testing the hypothesis $\beta = \beta_0,$ rests indeed on the assumption of normality for the variables $\epsilon_i$'s, so that, at least for limited values of $N,$ the whole procedure proves to be fairly ``unrobust" when such an assumption is not met (even if the asymptotic normality of $\hat{\beta}$ is assumed). The asymptotic theory for such intervals cannot be always invoked, besides, for such a theory rests on the asymptotic normality of $\hat{\beta}$ which is not always ensured ([1]).  

Two distinct methods can be used to solve the first of the  problems above: two distinct ways can be tried: one can decide to delete outliers or, alternatively, to base the estimation of $\beta$ on suitable functions of ranks, which are possibly unaffected by the extreme observations. Common thinking is that the deletion of outliers must follow rules  that are clearly stated before, and not after, data are available; this task cannot then rely on a subjective judgment, which will deprive the researcher of any foundation to study the related procedure. The papers by Brown and Mood ([2]), Adichie ([3]), Theil ([4]) and Sen ([5]) are framed, instead, in the logic of ranks, which proved to be able to overcome both the drawbacks  outlined above.  
 
 To introduce such kinds of procedures, notice that the estimator (\ref{uno}) can be rewritten so that the slopes
 $$ P_{ij} = \frac{Y_j-Y_i}{x_j-x_i}, \qquad i<j, $$
 are explicitly shown. Indeed, 
 $$ \hat{\beta} = \frac{\displaystyle{\sum_{i<j}} P_{ij} \, (x_j - x_i)^2}{\displaystyle{\sum_{i<j}} (x_j-x_i)^2} = \frac{\displaystyle{\sum_{i<j}} (Y_j-Y_i) \, (x_j - x_i)}{\displaystyle{\sum_{i<j}} (x_j-x_i)^2}. $$
 The above equality shows that $\hat{\beta}$ can be regarded as a mean of the $P_{ij}$'s with weights $(x_j-x_i)^2.$ To solve the problem of outliers, one can then obviously substitute such a weighted mean with a suitable function of the slopes $P_{ij},$ so as to result unaffected (at least less affected) by the extreme observations. This approach is substantially the one used by Theil, who proposed, as an estimator of $\beta,$ the median of the slopes $P_{ij}$, or the central value of the median interval when dealing with an even number of slopes. Theil's procedure is related to the one by Sen, who derived an estimator of $\beta$ by using a measure of concordance, which is essentially Kendall's $\tau,$ between the ranks of $Y_i-bx_i$ and those of $x_i,$ $i=1, 2, \ldots , N.$ The obtained estimator is the same proposed by Theil, but it can be applied under the general assumption that the $x_i$'s are not all distinct.
It is interesting to note that the same result can be obtained by starting from a completely different point of view, namely by using the minimax estimator with a non-quadratic loss function ([6]). 
The study of the asymptotic properties of both the point and the interval estimators is due to Sen as well, along with the determination of the asymptotic relative efficiency of the proposed estimator with respect to the one of least squares and to other estimators, proposed by Adichie ([3]), which were generalized, somehow under a more general framework, by Koul ([7]).
To have an idea of the efficiency gained by the Theil-Sen estimator, $\beta^*,$ with respect to that of least squares, $\hat{\beta},$ it suffices to notice that there are cases where
$$ \lim_{N \rightarrow +\infty} \frac{\Var (\hat{\beta})}{\Var (\beta^*)} = \alpha > 1 $$
and that, even in the normal case, if the constants $x_i$'s  are conveniently chosen,  
$$ \lim_{N \rightarrow +\infty} \frac{\Var (\hat{\beta})}{\Var (\beta^*)} = \frac{3}{\pi} \simeq 0.95. $$
  
Instead of measuring the concordance between the residuals $Y_i-bx_i$ and $x_i,$ $i=1,2,\ldots, N,$ by means of $\tau$ or other indices, as later proposed ([8]), one can obviously consider Gini's cograduation index $G.$ This procedure is quite different from the one proposed by Adichie, who used a class of indices which are functions of the ranks of residuals $Y_i-bx_i$ and of the values $x_i,$  while $G$ is based, as known, on the  ranks of $Y_i-bx_i$ and on the {\em ranks} of  $x_i,$ $i=1, 2, \ldots, N.$ In addition, the results gained  using $G$ are likely to be structurally different from the ones obtained from $\tau$ or Spearman's $R,$ because $G$ is believed to locate some aspects of cograduation which neither $\tau$ nor $R$ can account for. This statement, in effect, is also confirmed by the fact that the correlation coefficient between $G$ and $\tau$ (or between $G$ and $R$), as shown in ([9]) and in ([10]), even though quite large for a limited value of $N$ (in absence of cograduation), never reaches one, not even asymptotically. 
  
 Let $\underline y$ be a realization of $\underline Y,$ $G(\underline y,b)$ be Gini's cograduation index computed from the residuals $y_i-bx_i$ and $x_i,$ $i=1,2,\ldots, N,$ and $\tilde{\beta} (\underline y)$ be the function of data obtained by making $G(\underline y;b)$ as close to zero as possible.\footnote{As later shown, $G(\underline y;b)$ is actually a non increasing step function; hence the stated condition does not imply that $\tilde{\beta}$ is a root of the equation $G(\underline y;b)=0,$ which might not admit any root.} $\tilde{\beta}(\underline y)$ can then be regarded as a minimum $G-$dependence estimate or, more correctly, as a maximum $G-$indifference estimate of $\beta.$ The same notation can be used for the estimator $\tilde{\beta}(\underline Y).$ This terminology is coherent with the term "indifference" proposed by Gini ([11], p. 330) to indicate the lack of concordance or discordance between  two rankings, in comparison with the term ``(stochastic) independence" which should instead be used to indicate lack of connection. Indeed, the two conditions (independence and indifference) are not equivalent, though W. Hoeffding  ([12], p. 555) showed that, under suitable assumptions, they imply each other. 
      
This paper aims at proposing the estimator $\tilde{\beta}$ and at analyzing its properties. In section 2, the main problem is framed and the stochastic process $G(\underline Y;b),$ whose properties are studied in section 3, is introduced. In section 4 the estimator $\tilde{\beta}$ is formally defined and some properties of its distribution are analyzed. In section 5 the task of building a confidence interval for $\beta$ is faced, for every sample size and independently of the distribution function of disturbances, $F,$ which will be exclusively assumed to be continuous. As the proposed estimator does not possess a closed form as a function of data, section 6 gives some hints to fasten its computation for a given sample realization. Section 7 deals with the asymptotic distribution of $\tilde{\beta}.$ Such distribution is closely related to the one of $G(\underline Y;b)$ whose analysis is rather long and hence is developed in the Appendix, to simplify the structure of the paper. Finally, section 8 focuses on the comparison, based on the asymptotic relative efficiency (ARE), of the estimator $\tilde{\beta}$ with the one of least squares and the one by Theil and Sen. The drawn conclusions are quite interesting, as the asymptotic  efficiency of $\tilde{\beta} $  relative to the other two estimators is shown to be (for the chosen values of the $x_i$'s) greater than 1 when the distribution has  tails heavier than the normal case; this fact recommends a wide use of $\tilde{\beta}$, even if its computation might seem somehow unpractical. 
     
\section{Problem settings}     
Let $Y_1, Y_2, \ldots, Y_n$ be $N$ mutually independent random variables with distribution functions
$$ P\left\{ Y_i \leq y \right\} = F_i(y) = F(y- \alpha - \beta \, x_i) \qquad i=1,2, \ldots, N; \, N \geq 2 $$
where $F$ is any continuous distribution function and $x_1 < x_2 < \cdots < x_N$ are known constants. As the main interest is the estimation of $\beta,$ in the following $\alpha = 0$ will be supposed, without loss of generality.

For every real $b,$ consider the new variables
$$Z_i(b) = Y_i - b \, x_i \qquad i =1,2, \ldots , N$$
and use them to build the function (of $b$)
\begin{equation} \label{due}
G(\underline Y;b) = \frac 2D \sum_{1 \leq i \leq N} \left\{ \left| N+1-i -R(Z_i(b)) \right| - \left| i - R(Z_i(b)) \right| \right\} \: \: b \in \Re - B 
\end{equation} 
with 
$$ B = \left\{ b: \: \: Y_i - b \, x_i = Y_j  - b \, x_j  \: \: \mbox{ for at least a couple }   i \neq j \right\}, $$
where $R(Z_i(b))$ denotes the rank of $Z_i(b)$ in the sorting of $\{  Z_1(b), \ldots , Z_N(b) \}$ and $D=N^2$ if $N$ is even or $D=N^2-1$ if $N$ is odd.      
     
The function $G(\underline Y;b)$  is not defined in the set $B,$ which is clearly finite. For every chosen $b \notin B,$ the function (\ref{due}) is the known Gini's cograduation index between $Y_1- b \, x_1, \ldots , Y_N - b \, x_N$ and $x_1, \ldots , x_N;$ conversely, as a function of $b,$ it is a stochastic process whose realizations correspond to the events $\omega \equiv (y_1, \ldots , y_N) \in A \subseteq \Re^N.$
     
As the random variables $Y_i - \beta \, x_i ,$ $i=1,2, \ldots, N,$  are iid, $G(\underline Y;\beta)$ has the known distribution of Gini's index in case of indifference and hence
$$ \E \left\{ G(\underline Y; \beta) \right\} = 0. $$  
One can then naturally estimate the parameter $\beta$ by making $G(\underline Y;b)$  as close to zero as possible, that is by letting
$$ G (\underline Y; b) \simeq 0 .$$
 Equivalently, the estimator proposed in this paper is a function $\tilde{\beta} = \tilde{\beta}(\underline Y)$ so that the sequence $Y_i-\tilde{\beta} \, x_i$ will result  as indifferent as possible to $x_i,$ $i=1, 2,\ldots , N.$ In effect, this is a natural requirement when considering that the least squares estimator can be regarded as a function $\hat{\beta} = \hat{\beta} (\underline Y)$ which makes the usual sample covariance between $Y_i - \hat{\beta} \, x_i$ and $x_i$ 
$$ \frac 1N \sum_{1 \leq i \leq N} \left( Y_i - \bar Y - \hat{\beta} \, (x_i - \bar x) \right) \, (x_i - \bar x) $$
vanish. Such a covariance plays then, in another framework, the same role of $G(\underline Y;\tilde{\beta}).$  
 
Of course, to implement the proposed procedure one must be sure that the obtained estimator is, in some sense, unique. This could be the case if the realizations of the process $G(\underline Y;b)$ resulted  strictly monotonic functions of $b.$ In the following section, such realizations are shown to be non increasing functions of $b.$ This fact implies that  a whole interval of values of $b$ may exist where $G(\underline Y; b)=0$ or, alternatively, two consecutive intervals $I_1$ and $I_2$ so that
$$ G(\underline Y;b)>0 \quad b \in I_1 \qquad \mbox{and} \qquad G(\underline Y;b)<0 \quad b \in I_2 .$$

\section{Properties of $G(\underline Y;b)$}

As claimed in the previous section, $G(\underline Y;b)$ is not defined for every real $b;$ more specifically, if $\underline y = (y_1, \ldots, y_N)$ is a realization of $\underline Y,$ $G(\underline y;b)$ turns out to be undefined in the set
$$ B = \left\{ b: \: \: y_i - b \, x_i = y_j  - b \, x_j  \: \: \mbox{ for at least a couple }   i \neq j \right\}, $$  
which will be referred to, in the following, as
$$ B = \left\{ b^{(1)}, b^{(2)}, \cdots , b^{(r)} \right\} \qquad 1 \leq r \leq {N \choose 2} \qquad b^{(1)} < b^{(2)} < \cdots < b^{(r)}.$$   
   
For any $n-$tuple $\underline y,$ the function $G(\underline y;b)$ is constant inside each interval
$$(-\infty,b^{(1)}), \: (b^{(1)},b^{(2)}), \: \cdots  \:, \: (b^{(r)}, +\infty) .$$
To prove such a claim, it suffices to show that, inside each of the  intervals above, $\{ R(Z_1(b)), \ldots , R_N(Z(b)) \}$ is the same permutation of the set of integers $\{ 1, 2,$  $\ldots , N \}.$

Let $b^{(i)} < b < b^{(i+1)},$ $i=1,\ldots, r-1.$ There are at least two couples of indices $(u_1,v_1)$ and $(u_2,v_2),$ with $u_1<v_1,$ $u_2<v_2,$ such that 
$$ b^{(i)} = \frac{y_{v_1} - y_{u_1}}{x_{v_1} - x_{u_1}}  \, < \, b \, < \, \frac{y_{v_2} - y_{u_2}}{x_{v_2} - x_{u_2}}  = b^{(i+1)} .$$
This fact implies that, for every $b$ belonging to the interval $\left( b^{(i)}, b^{(i+1)} \right),$
$$ R(Z_{v_1}(b)) < R(Z_{u_1}(b)) \qquad R(Z_{v_2}(b)) > R(Z_{u_2}(b)). $$
The former of the above inequalities holds equivalently for every couple whose slope is less than or equal to $b^{(i)}$; the latter inequality holds for those couples whose slope is greater  than or equal to $b^{(i+1)}.$ This remark shows that, for every $b$ belonging to the considered interval, the permutation taken by $\{ R(Z_1(b)), \ldots ,$ $R_N(Z(b)) \}$ does not change, which suffices to state that $G(\underline y;b)$ does not change its value. The same conclusions can be drawn when considering the first and the last intervals for $b.$ Specifically, as
$$ 
\begin{array}{rcll}
\{ R(Z_1(b)), \ldots , R_N(Z(b)) \} &=& \{ 1,2, \ldots , N \} & \quad \forall b < b^{(1)} \\
\{ R(Z_1(b)), \ldots , R_N(Z(b)) \} &=& \{ N, N-1,\ldots , 1 \} & \quad \forall b > b^{(r)} ,
\end{array} $$
one gets
$$ G(\underline y;b) = 1 \: \: \forall b < b^{(1)} \qquad \mbox{and} \qquad G(\underline y;b) = -1 \: \: \forall b > b^{(r)} .$$

Obviously the definition of $G(\underline Y;b)$ can be supplemented by setting
$$ G(\underline y; b^{(i)}) = \lim_{b \rightarrow b^{(i)+}} \, \, G(\underline y;b) ,$$
so as to let every realization of the process be right continuous. In the following, $G(\underline Y;b)$ will be supposed to be defined for every real $b.$

Consider now two adjacent intervals ${\cal I}_1 = [ b^{(s)}, b^{(s+1)} )$ and ${\cal I}_2 = [ b^{(s+1)},$ $ b^{(s+2)} )$ and let $b \in {\cal I}_1,$ $b_1 \in {\cal I}_2.$ When shifting from $b$ to $b_1,$ the above discussion shows that the ranks of $Z_i(b)$ will be only partially modified. Specifically, suppose that 
$$ b^{(s+1)} = \frac{y_{v_1} - y_{v_0}}{x_{v_1} - x_{v_0}} =  \frac{y_{v_2} - y_{v_0}}{x_{v_2} - x_{v_0}} = \cdots =  \frac{y_{v_m} - y_{v_0}}{x_{v_m} - x_{v_0}} $$
with 
$$v_0 < v_1 < \cdots <v_m \qquad \mbox{and} \qquad m \geq 1 ,$$
which means that the observations $y_{v_0}, y_{v_1} , \ldots , y_{v_m}$ lie on the same straight line. When shifting from $b$ to $b_1,$ only the ranks of $Z_{v_0}, \ldots , Z_{v_m}$ will be modified, that is
\begin{equation} \label{tre}
R(Z_{v_i}(b_1)) = R(Z_{v_i}(b)) + m -2i \qquad i =0,1,\ldots, m;
\end{equation}
furthermore
\begin{equation} \label{quattro}
R(Z_{v_i}(b)) = R(Z_{v_0}(b)) +i \qquad i =0,1,\ldots, m.
\end{equation}
The above equalities derive immediately after considering that the rank of the generic $Z_j(b)$ equals the number of observations $y_j$ which lie under or on the straight line with slope $b$ passing through $(x_j,y_j).$ 

\begin{lemma*}
If $1 \leq v_0 < v_1 < \cdots < v_m,$ $m \geq 0,$ then
$$ \sum_{ 1 \leq i \leq m} \left| v_i - m +i - \xi \right| \:  \geq \: \sum_{1 \leq i \leq m} \left| v_i - i - \xi \right| \qquad \forall \xi \in \Re. $$
\end{lemma*}

\begin{proof}
Consider the functions
$$ \phi_i (\xi) = \left| v_i - m + i -\xi \right| - \left| v_i - i - \xi \right| \qquad i =0,1, \ldots , m $$
and define
$$ \phi (\xi) = \sum_{0 \leq i \leq m} \phi_i(\xi) =  \sum_{0 \leq i \leq [m/2]} \left( \phi_i(\xi) + \phi_{m-i} (\xi) \right). $$
For every $0 \leq i \leq [m/2],$ one gets
$$ \begin{array}{rcl}
\phi_i(\xi) &=& 
\left\{ \begin{array}{l@{\qquad}l}
-(m-2i) \leq 0 & \xi< v_i -m+i \\
2\xi-2v_i+m & v_i -m +i \leq \xi <v_i -i \\
m-2i \geq 0 & \xi \geq v_i -i 
\end{array} \right. \\ \\
\phi_{m-i} (\xi) &=&
\left\{ \begin{array}{l@{\qquad}l}
m -2i \geq 0 & \xi < v_{m-i} - m + i \\
2v_{m-i} - m - 2 \xi & v_{m-i} - m + i \leq \xi < v_{m-i} - i \\
-(m-2i) \leq 0 & \xi \geq v_{m-i} -i 
\end{array} \right.
\end{array} $$
hence
$$ \phi_i(\xi) + \phi_{m-i} (\xi) \geq 0 \qquad \forall \xi \in \Re, $$
from which the proof follows. $\qed$
\end{proof}

The following theorem can now be stated.

\begin{theorem}
For every $n-$tuple $\underline y,$ the function $G(\underline y;b)$ is non increasing.
\end{theorem}

\begin{proof}
It suffices to prove that the function
$${\cal U}(b) = \sum_{1 \leq i \leq N} \left| i - R(Z_i(b)) \right| $$
is non decreasing and that the function
$${\cal V}(b) = \sum_{1 \leq i \leq N} \left| N+1 - i - R(Z_i(b)) \right| $$
is non increasing. Only the statement for ${\cal U}(b)$ will be proved; the one for ${\cal V}(b)$ follows similarly. Suppose that $b^{(s)} < b < b^{(s+1)},$ $b^{(s+1)} < b_1 < b^{(s+2)},$ with $b^{(0)} = -\infty,$ $b^{(r+1)} = + \infty,$
$$
b^{(s+1)} = \frac{y_{v_{i,t}} - y_{v_{0,t}}}{x_{v_{i,t}} - x_{v_{0,t}}} \qquad i = 1,2,\ldots, m_t \qquad t=1,2,\ldots ,k \geq 1 $$
and $v_{0,t} < v_{1,t} < \cdots < v_{m,t}.$ When shifting from $b$ to $b_1,$ only the ranks of
\begin{align*} 
\left\{ Z_{v_{0,1}}(b), \ldots , Z_{v_{m_1,1}}(b) \right\} \: , \:
\left\{ Z_{v_{0,2}}(b), \ldots , \right. & \left. Z_{v_{m_2,2}}(b) \right\}, \:   \ldots \: , \\
& \left\{ Z_{v_{0,k}}(b), \ldots , Z_{v_{m_k,k}}(b) \right\} 
\end{align*} 
will change; hence 
\begin{eqnarray*}
\lefteqn{{\cal U}(b_1) -  {\cal U}(b) = } \\
&& = \sum_{1 \leq t \leq k} \left\{ \sum_{0 \leq i \leq m_t}
\left| v_{i,t} - R(Z_{v_{i,t}} (b_1)) \right| - 
\sum_{0 \leq i \leq m_t}
\left| v_{i,t} - R(Z_{v_{i,t}} (b)) \right|
\right\} .
\end{eqnarray*}
By (\ref{tre}) and (\ref{quattro}),
\begin{align*}
 {\cal U}(b_1) -  {\cal U}(b) & =  
\sum_{1 \leq t \leq k} \left\{ \sum_{0 \leq i \leq m_t}
\left| v_{i,t} - m_t + i - R(Z_{v_{0,t}} (b)) \right| + \right. \\
& \quad \left. - \sum_{0 \leq i \leq m_t}
\left| v_{i,t} - i - R(Z_{v_{0,t}} (b)) \right|
\right\} .
\end{align*}
According to the Lemma stated above, the quantity in brackets is non negative for every value of $R(Z_{v_{0,t}}(b))$ and hence
$${\cal U}(b_1) \geq {\cal U}(b), $$
which holds for all intervals and thus gives the proof. $\qed$
\end{proof}

\section{Definition of the point estimator and related properties}
Section 3 showed that all trajectories of the stochastic process $G(\underline Y;b)$ are non increasing functions of $b$ and that
$$ G(\underline y;b) = 1 \quad \forall b < b^{(1)}; \qquad  \qquad G(\underline y;b) = -1 \quad \forall b > b^{(r)} .$$
For any observed $n-$tuple $\underline y,$ the following two cases will then arise:
\begin{itemize}
\item[a)]
 a whole interval of $b$ exists where $G(\underline y;b)=0$ pointwise
 \item[b)]
  two adjacent intervals exist such that
 $$\begin{array}{rl}
 G(\underline y;b) >0 & \qquad  \mbox{in the first interval} \\
 G(\underline y;b) <0 & \qquad  \mbox{in the second interval.}
 \end{array} $$
\end{itemize}
When in case a), one could choose the central value of the interval as an estimate of $\beta;$ in case b), the value
$$ \sup \left\{ b: \: \: G(\underline y;b)>0 \right\} = \inf \left\{ b: \: \: G(\underline y;b)<0 \right\} $$
could instead be chosen. 
The two cases can then be unified by defining the following maximum $G-$indifference estimator:
\begin{equation} \label{cinque}
\tilde{\beta} = \tilde{\beta}(\underline y) = \frac 12 \left[ \sup \left\{ b: \:  G(\underline Y;b)>0 \right\} + \inf \left\{ b:  \: G(\underline Y;b)<0 \right\} \right]
\end{equation}
which is similar to the estimator proposed in ([13]) for the location parameter. One of the following sections will show how to get a fast computation of $\tilde{\beta}(\underline y).$ First of all, the next propositions will give three quite immediate properties of the distribution of $\tilde{\beta}.$  

\begin{proposition}
The distribution of $\tilde{\beta} - \beta$ does not depend on the parameter $\beta.$
\end{proposition}

\begin{proof}
Let 
$$ \tilde{\beta}_1 (Y_1-\beta \, x_1, \ldots , Y_N-\beta \, x_N) = \sup \left\{ b: \: G(Y_1 - \beta \, x_1, \ldots , Y_N - \beta \, x_N; \, b) >0 \right\} .$$ 
From the definition of $G(\underline Y;b),$
$$ \tilde{\beta}_1 (Y_1-\beta \, x_1, \ldots , Y_N-\beta \, x_N)  = \tilde{\beta}_1(\underline Y) - \beta .$$
Similarly, if 
$$ \tilde{\beta}_2 (Y_1-\beta \, x_1, \ldots , Y_N-\beta \, x_N) = \inf \left\{ b: \: G(Y_1 - \beta \, x_1, \ldots , Y_N - \beta \, x_N; \, b) <0 \right\} ,$$
one gets
$$ \tilde{\beta}_2 (Y_1-\beta \, x_1, \ldots , Y_N-\beta \, x_N)  = \tilde{\beta}_2(\underline Y) - \beta $$
and, by definition (\ref{cinque}),
$$ \tilde{\beta}(Y_1-\beta \, x_1, \ldots , Y_N-\beta \, x_N)  = \tilde{\beta}(\underline Y) - \beta .$$
However, the lhs of the above equality is a function of the variables $Y_1-\beta \, x_1, \ldots , Y_N-\beta \, x_N,$ whose distributions, by hypothesis, do not depend on $\beta.$ $\qed$
\end{proof}

Proposition 1 equivalently states that, if $P\left( \tilde{\beta} \leq b \right) = \psi(b; \beta),$ then $\psi(b; \beta)=\varphi(b-\beta),$ namely $\beta$ is a location parameter of the distribution of $\tilde{\beta}.$ This fact will allows setting $\beta = 0$ in the following, without loss of generality. 

\begin{proposition}
$\tilde{\beta}$ has a continuous distribution.
\end{proposition}

\begin{proof}
It suffices to prove that the two variables
\begin{align*}
\tilde{\beta}_1 (\underline Y) = & \sup \{ b: G(\underline Y;b) >0 \} \\
\tilde{\beta}_2 (\underline Y) = & \inf \{ b: G(\underline Y;b) <0 \} 
\end{align*}
are both continuous. The continuity of $\tilde{\beta}_1$ and $\tilde{\beta}_2,$ indeed, implies that the joint distribution of $(\tilde{\beta}_1, \tilde{\beta}_2)$ is continuous and similarly for $\tilde{\beta} = \frac 12 \left( \tilde{\beta}_1 + \tilde{\beta}_2 \right).$

As every realizations of $G(\underline Y;b)$ is non-increasing with at least a jump at a point of the form $(Y_j-Y_i)/(x_j-x_i),$ $i \neq j,$ the event $\tilde{\beta} = a \in \Re$ implies that $G(\underline Y;b)$ has a jump at $a.$ Hence
\begin{align*}
P \left\{ \tilde{\beta}_1 = a \right\} & \leq  P \left\{ \frac{Y_j-Y_i}{x_j-x_i} = a \quad \mbox{ for at least a couple } (i<j) \right\} \\
& \leq  \sum_{i<j} \, P \left\{ \frac{Y_j-Y_i}{x_j-x_i} = a  \right\} 
\end{align*}
As, by hypothesis, the variables $Y_k$'s are continuous (and independent), the same is true for the variables $(Y_j-Y_i)/(x_j-x_i);$ hence $P(\tilde{\beta}_1 = a) =0.$ Similarly, one can show that  $P(\tilde{\beta}_2 = a) =0.$   $\qed$
\end{proof}

Notice that, by following the same steps as for the proof of Proposition 2, a similar result can be obtained for the variables
\begin{align*}
\tilde{\beta}_I (\underline Y) = & \inf \{ b: G(\underline Y;b) <G^* \} \\ 
\tilde{\beta}_S (\underline Y) = & \sup \{ b: G(\underline Y;b) >-G^*\} 
\end{align*}
where $G^*>0$ is a given constant.

Before stating another property concerning the distribution of $\tilde{\beta},$ the following equality should be considered:
\begin{equation} \label{sei}
G(-\underline y;b) = - G(\underline y;-b) \qquad \forall \, \underline y \subseteq \Re^N \mbox{ and } \forall \, b \in \Re. 
\end{equation}
Indeed, 
\begin{align*}
G(- \underline y;b) & = \frac 2D \, \left\{ \sum_i \left| N+1-i-R(-y_i+b\, x_i) \right| - \sum_i \left| i - R(-y_i+b\, x_i) \right| \right\} \\
& = \frac 2D \, \left\{ \sum_i \left| N+1-i-R[-(y_i-b\, x_i)] \right| - \sum_i \left| i - R[-(y_i-b\, x_i)] \right| \right\}  
\end{align*}
from which the above result follows, as it is obviously
$$R[-(y_i-b\, x_i)] = N+1 -  R(y_i-b\, x_i). $$

\begin{proposition}
If $Y_1, Y_2, \ldots , Y_n$ are symmetrically distributed, 
$$ \E (\tilde{\beta}) = \beta \qquad \forall \, N \geq 2 \quad \forall \, \beta \in \Re .$$
\end{proposition}

\begin{proof}
According to Proposition 1, it can be assumed that $\beta =0.$ Now notice that
\begin{equation} \label{sette}
\tilde{\beta} (- \underline Y ) = - \tilde{\beta} (\underline Y)
\end{equation}
Indeed, 
$$ \tilde{\beta} (- \underline Y ) = \frac 12 \left[ \sup \left\{ b: \, G(-\underline Y;b) >0 \right\} + 
 \inf \left\{ b: \, G(-\underline Y;b) <0 \right\} \right] $$
 and, by (\ref{sei}),
 \begin{align*}
 \tilde{\beta} (- \underline Y ) &= \frac 12 \left[ \sup \left\{ b: \, G(\underline Y;-b) <0 \right\} + 
 \inf \left\{ b: \, G(\underline Y;-b) >0 \right\} \right] = \\
 &= \frac 12 \left[ -\inf \left\{ b: \, G(\underline Y;b) <0 \right\} - 
 \sup \left\{ b: \, G(\underline Y;b) >0 \right\} \right] = \\
 & = -\tilde{\beta} ( \underline Y )
 \end{align*}
 By the symmetry of $Y_1, Y_2, \ldots , Y_N,$ the variables $\tilde{\beta} ( \underline Y )$ and $\tilde{\beta} ( -\underline Y )$ share the same distribution. By using (\ref{sette}) and this latter property, one can then claim that $\tilde{\beta} ( \underline Y )$ has a distribution symmetric around zero (which is the value of $\beta$); this fact completes the proof. $\qed$
 \end{proof}

\section{Confidence intervals for $\beta$}
As the variables $Y_i - \beta \, x_i,$ $i=1,2, \ldots , N$ are iid, $G(\underline Y ; \beta)$ has the known distribution of Gini's cograduation index under indifference. There exist quite complete tables of such a distribution.  By using these tables, a constant $G^*>0$ such that, for a suitable $\alpha,$
$$ P \{ - G^* < G(\underline Y;\beta) < G^* \} = 1-\alpha \qquad 0 < \alpha < 1 $$
can be easily determined. Consider now the variables
\begin{align}
\tilde{\beta}_I (\underline Y) = & \inf \{ b: G(\underline Y;b) <G^* \} \label{otto} \\ 
\tilde{\beta}_S (\underline Y) = & \sup \{ b: G(\underline Y;b) >-G^* \label{nove} \} .
\end{align}
From (\ref{otto}), as $G(\underline y;b)$ is non increasing, 
$$ \inf \{ b: \, G(\underline y;b) < G^* \} < \beta \, \, \Rightarrow \, \, G(\underline y;\beta) < G^* \,\, \Rightarrow \, \, 
\inf \{ b: \, G(\underline y;b) < G^* \} \leq \beta $$
Similarly, from (\ref{nove}),
$$ \sup \{ b:  G(\underline y;b) > -G^* \} > \beta   \Rightarrow \,  G(\underline y;\beta) > - G^* \Rightarrow \,  \sup \{ b:  G(\underline y;b) > -G^* \} \geq \beta $$
It follows that
\begin{equation} \label{dieci}
\{ \tilde{\beta}_I < \beta <  \tilde{\beta}_S \} \,\,   \Rightarrow \, \, \{ - G^* <  G(\underline Y;\beta) <  G^* \} \, \, \Rightarrow \, \,   \{ \tilde{\beta}_I \leq \beta \leq  \tilde{\beta}_S \} .
\end{equation}
Hence, from (\ref{dieci}), 
$$ P\{ \tilde{\beta}_I < \beta <  \tilde{\beta}_S \} \,\leq  \, P \{ - G^* <  G(\underline Y;\beta) <  G^* \} \, \leq  \,  P \{ \tilde{\beta}_I \leq \beta \leq  \tilde{\beta}_S \} $$
and, by the continuity of $\tilde{\beta}_I$ and $\tilde{\beta}_S,$
$$ P\{ \tilde{\beta}_I < \beta <  \tilde{\beta}_S \} \,=  \, P \{ - G^* <  G(\underline Y;\beta) <  G^* \} \, =   \,  1-\alpha .$$
The variables (\ref{otto}) and (\ref{nove}) are then respectively the lower and the upper bounds of the confidence interval for $\beta,$ for any continuous distribution function $F.$ 

\section{Computation of $\tilde{\beta}$}
In the previous sections, the point estimator for $\beta$ and the bounds of the confidence interval for the same parameter were defined. However, a closed expression for such statistics as functions of the elements of the sample, was not provided.  This fact makes it difficult to study further properties of the considered statistics for a finite value of $N.$ The following section will then deal with the asymptotic distribution of $\tilde{\beta}.$ Before doing that, this section aims at providing an easy scheme to determine the values taken by $\tilde{\beta},$ $\tilde{\beta}_I$ and $\tilde{\beta}_S$ for any given sample realization. 

Let $\underline y = (y_1, \ldots , y_N)$ denote the observations on the response variable corresponding to $x_1, \ldots , x_N$ and suppose computing the ${N \choose 2}$ slopes $P_{ij} = \frac{y_j-y_i}{x_j-x_i},$ $i<j$ which are not necessarily all distinct. Denote with ${}_{(k)}P_{ij}$ the distinct sorted values of such slopes, $k=1, 2, \ldots, r.$ Of course the same slope can correspond to more than a couple of indices $(i,j).$   With the aid of equality (\ref{tre}) in section 3, one can then produce a table displaying, for each row $i=1,2,\ldots, N,$ the ranks of $y_i-b\, x_i$ when $b$ belongs to the possible intervals determined by the slopes ${}_{(k)}P_{ij}.$ To illustrate the construction of such a table, suppose that the observed values are
$$ \begin{array}{c|@{\quad}c@{\quad}c@{\quad}c@{\quad \:}c} 
x_i & 1 & 2 & 3 & 4 \\ \hline
y_i & 2 & 2.5 & 4 & 5 
\end{array} $$
The six possible slopes are
$$ P_{12}=0.5 \quad P_{13}=1 \quad P_{14}=1 \quad P_{23}=1.5 \quad P_{24}=1.25 \quad P_{34}=1, $$
so that the sorted distinct slopes are
$$ {}_{(1)}P_{12} = 0.5 \quad  {}_{(2)}P_{13} =  {}_{(2)}P_{14} =  {}_{(2)}P_{34} = 1 \quad 
 {}_{(3)}P_{24} = 1.25 \quad  {}_{(4)}P_{23} = 1.5. $$
A table with $N=4$ rows can now be produced as follows. First of all a vertical line is built for every slope ${}_{(k)}P_{ij}$ and, on this line, the $i-$th row is marked with a circle and the $j-$th row is marked with a square. For every $h-$th row, one can then put suitable integer values,  starting from $h,$ by adding a unit if a circle is met and by subtracting a unit if a square is met. If more than a single circle or square is met, the number in the previous column will be simply increased by the number of circles and decreased by the number of squares. As an example, on the third row, when passing from the second to the third column, a circle and a square are met; the number on the second column (3) should then  be increased by 1 and decreased by 1, so that   the same value (3) is reported in the third column. 

\begin{center}
\setlength{\unitlength}{1cm}
\begin{picture}(8,8)(0,0)
\put(1,4.1){\line(1,0){5}}
\put(1,5.1){\line(1,0){5}}
\put(1,6.1){\line(1,0){5}}
\put(1,7.1){\line(1,0){5}}
\put(2,3.5){\line(0,1){4.5}}
\put(3,3.5){\line(0,1){4.5}}
\put(4,3.5){\line(0,1){4.5}}
\put(5,3.5){\line(0,1){4.5}}
\put(0.5,4){4}
\put(0.5,5){3}
\put(0.5,6){2}
\put(0.5,7){1}
\put(6.5,4.5){$R(Z_4(b))$}
\put(6.5,5.5){$R(Z_3(b))$}
\put(6.5,6.5){$R(Z_2(b))$}
\put(6.5,7.5){$R(Z_1(b))$}
\put(1.5,4.5){4}
\put(1.5,5.5){3}
\put(1.5,6.5){2}
\put(1.5,7.5){1}
\put(2.5,4.5){4}
\put(2.5,5.5){3}
\put(2.5,6.5){1}
\put(2.5,7.5){2}
\put(3.5,4.5){2}
\put(3.5,5.5){3}
\put(3.5,6.5){1}
\put(3.5,7.5){4}
\put(4.5,4.5){1}
\put(4.5,5.5){3}
\put(4.5,6.5){2}
\put(4.5,7.5){4}
\put(5.5,4.5){1}
\put(5.5,5.5){2}
\put(5.5,6.5){3}
\put(5.5,7.5){4}
\put(1.5,3){$\begin{array}{c} {}_{(1)}P_{12} \\ \scriptstyle{\bf (0.5)} \end{array}$}
\put(2.5,2.6){$\begin{array}{c} {}_{(2)}P_{13} \\ {}_{(2)}P_{14} \\ {}_{(2)}P_{34} \\ \scriptstyle{\bf (1)} \end{array}$}
\put(3.5,3){$\begin{array}{c} {}_{(3)}P_{24} \\ \scriptstyle{\bf (1.25)} \end{array}$}
\put(4.5,3){$\begin{array}{c} {}_{(4)}P_{23} \\  \scriptstyle{\bf (1.5)} \end{array}$}
\put(2,7.1){\circle{0.2}}
\put(3,7.1){\circle{0.2}}
\put(3,7.1){\circle{0.35}}
\put(4,6.1){\circle{0.2}}
\put(5,6.1){\circle{0.2}}
\put(3,5.1){\circle{0.35}}
\put(1.9,6.0){\dashbox{0.2}(0.2,0.2)}
\put(2.9,5.0){\dashbox{0.2}(0.2,0.2)}
\put(4.9,5.0){\dashbox{0.2}(0.2,0.2)}
\put(2.9,4.0){\dashbox{0.2}(0.2,0.2)}
\put(2.825,3.925){\dashbox{0.35}(0.35,0.35)}
\put(3.9,4.0){\dashbox{0.2}(0.2,0.2)}
\end{picture}
\end{center}

\vspace{-2.2cm}

The figures on the $h-$th row of the above table are the ranks of $y_h-b\, x_h$ as long as $b$ ranges in the intervals $(-\infty, {}_{(1)}P),$ $({}_{(1)}P, {}_{(2)}P),$ $({}_{(2)}P, {}_{(3)}P),$ $({}_{(3)}P, {}_{(4)}P),$ $({}_{(4)}P, + \infty).$ For example, 
$$ \begin{array}{ll@{\qquad}rl@{\,}l}
R(Z_2(b)) & = 2 & &b < {}_{(1)}P & = 0.5 \\
& = 1 & 0.5 \leq  &b < {}_{(3)}P & = 1.25 \\
& =2 & 1.25 \leq  &b < {}_{(4)}P & = 1.5 \\
& =3 &&b \geq 1.5 
\end{array} $$
After the above table, two further tables can be produced by computing, for the $i-$th row, the quantities 
$$ |N+1-i-R(Z_i(b))| \qquad \mbox{ and } \qquad |i-R(Z_i(b))|$$
One then gets

\vspace{-0.5cm}

\begin{center}
\setlength{\unitlength}{1cm}
\begin{picture}(12,3.5)(0,0)
\put(1,0.7){\line(1,0){5}}
\put(1,1.4){\line(1,0){5}}
\put(1,2.1){\line(1,0){5}}
\put(1,2.8){\line(1,0){5}}
\put(2,0){\line(0,1){3.5}}
\put(3,0){\line(0,1){3.5}}
\put(4,0){\line(0,1){3.5}}
\put(5,0){\line(0,1){3.5}}
\put(1.5,0.2){\bf 8}
\put(1.5,0.9){3}
\put(1.5,1.6){1}
\put(1.5,2.3){1}
\put(1.5,3.0){3}
\put(2.5,0.2){\bf 8}
\put(2.5,0.9){3}
\put(2.5,1.6){1}
\put(2.5,2.3){2}
\put(2.5,3.0){2}
\put(3.5,0.2){\bf 4}
\put(3.5,0.9){1}
\put(3.5,1.6){1}
\put(3.5,2.3){2}
\put(3.5,3.0){0}
\put(4.5,0.2){\bf 2}
\put(4.5,0.9){0}
\put(4.5,1.6){1}
\put(4.5,2.3){1}
\put(4.5,3.0){0}
\put(5.5,0.2){\bf 0}
\put(5.5,0.9){0}
\put(5.5,1.6){0}
\put(5.5,2.3){0}
\put(5.5,3.0){0}
\put(0.5,0.2){\bf Tot.}
\put(7,0.7){\line(1,0){5}}
\put(7,1.4){\line(1,0){5}}
\put(7,2.1){\line(1,0){5}}
\put(7,2.8){\line(1,0){5}}
\put(8,0){\line(0,1){3.5}}
\put(9,0){\line(0,1){3.5}}
\put(10,0){\line(0,1){3.5}}
\put(11,0){\line(0,1){3.5}}
\put(7.5,0.2){\bf 0}
\put(7.5,0.9){0}
\put(7.5,1.6){0}
\put(7.5,2.3){0}
\put(7.5,3.0){0}
\put(8.5,0.2){\bf 2}
\put(8.5,0.9){0}
\put(8.5,1.6){0}
\put(8.5,2.3){1}
\put(8.5,3.0){1}
\put(9.5,0.2){\bf 6}
\put(9.5,0.9){2}
\put(9.5,1.6){0}
\put(9.5,2.3){1}
\put(9.5,3.0){3}
\put(10.5,0.2){\bf 6}
\put(10.5,0.9){3}
\put(10.5,1.6){0}
\put(10.5,2.3){0}
\put(10.5,3.0){3}
\put(11.5,0.2){\bf 8}
\put(11.5,0.9){3}
\put(11.5,1.6){1}
\put(11.5,2.3){1}
\put(11.5,3.0){3}
\put(6.5,0.2){\bf Tot.}
\end{picture}
\end{center}


The total of every column in the left table above gives the value of 
$$\sum_{1  \leq i \leq N}|N+1-i-R(Z_i(b))|$$
 when $b$ ranges in each interval; similarly, the totals in the right table give the values of 
 $$\sum_{1  \leq i \leq N}|i-R(Z_i(b))|.$$
 Such totals provide an easy computation of the value taken by $\tilde{\beta}.$ Indeed, after noticing that in the  considered example $D=8,$ one gets
 \begin{equation} \label{undici}
\begin{array}{ll@{\qquad}rl@{\,}ll}
G(\underline y;b) & = \frac{8-0}{8} = 1  & &b < 0.5 &\, =\,  b^{(1)} \vspace{0.2cm}\\
& = \frac{8-2}{8} = \frac 34 & 0.5 \leq  &b < 1  & \, =\,   b^{(2)}  \vspace{0.2cm} \\ 
& = \frac{4-6}{8} = -\frac 14 & 1 \leq  &b < 1.5  & \, =\,   b^{(3)}  \vspace{0.2cm} \\
& = \frac{2-6}{8} = -\frac 12 & 1.25 \leq  &b < 1.5  & \, =\,   b^{(4)}  \vspace{0.2cm} \\
& = \frac{0-8}{8} = -1  &  &b \geq  1.5   
\end{array} 
\end{equation}
so that $\tilde{\beta}=1.$ Notice that $\tilde{\beta}=1$ is also the median of the possible slopes, even if this coincidence is not a general rule. The least-squares estimate is $\hat{\beta}=1.05$ instead. 

Concerning the determination of the confidence interval for $\beta$ and thus of the bounds $\tilde{\beta}_I$ and $\tilde{\beta}_S,$ notice that the tables of the distribution of $G$ under indifference provide
$$ P \{ -1 <G(\underline Y;\beta) < 1 \} = P \left\{ -\frac 34 \leq G(\underline Y;\beta) \leq \frac 34 \right\} = \frac{11}{12} \simeq 0.92. $$
By (\ref{undici}), one can then deduce that
$$\begin{array}{rll}
\tilde{\beta}_I  &= \inf \{ b: \, G(\underline y;b) < 1 \}  &= 0.5  = \inf \left\{ b: \, G(\underline y;b) \leq \frac 34 \right\}  \vspace{0.2cm}\\ 
\tilde{\beta}_S  &= \sup \{ b: \, G(\underline y;b) >- 1 \}  &= 1.5  = \sup \left\{ b: \, G(\underline y;b) \geq - \frac 34 \right\}  
\end{array} $$
so that the confidence interval for $\beta$ with level $1-\alpha =92 \%,$ whatever the distribution function $F,$ is 
$$ 0.5 < \beta < 1.5.$$
Notice that the least squares method cannot provide a similar result, without any further assumptions.

\section{The asymptotic distribution of $\tilde{\beta}$} \label{s7}
In order to compare the estimator $\tilde{\beta}$ with the other cited estimators for $\beta,$ some information about its asymptotic distribution is needed. The following theorem, whose proof is found in the Appendix, will be of use 
\begin{theorem} \label{t2}
Let $Y_1, Y_2, \ldots, Y_n $ be independent variables with a common distribution function $F$ and absolutely continuous density $f,$ whose support is $\Re,$ and suppose that
\begin{enumerate}
\item[i)] $I(f) = \int_{-\infty}^{+\infty} \left( \frac{f'}{f} \right)  \, f \, dy < \infty $
\item[ii)] $\int_{-\infty}^{+\infty} |f'| \, dy < \infty $
\end{enumerate}
and that $\displaystyle{\frac{T^2}{M} = \frac 1M \sum_{i=1}^N (x_i - \bar x)^2 \rightarrow +\infty,}$ with $M = \displaystyle{\max_{1 \leq i \leq N} (x_i -\bar x)^2,}$ then
$$ \lim_{N \rightarrow + \infty} P \left\{ \sqrt N \, G \left( \underline Y; \frac bT \right) \leq z \right\} = \phi \left( \frac{z-\sigma_{12}}{\sqrt{2/3}} \right) $$
where $\phi$ denotes the normal cdf with zero mean and unit variance and
\begin{align*}
\sigma_{12} &= 4b \, \int_{-\infty}^{+\infty} \left[ \psi(1-F(y)) - \psi(F(y)) \right] \, f'(y) \, dy \\
\psi(y) & = \lim_{N \rightarrow + \infty} \, \frac{1}{N^{3/2} T} \, \sum_{i=1}^{[Ny]} (x_i - \bar x) (Ny-i) \qquad 0 <y \leq 1
\end{align*}
\end{theorem}

\noindent
{\em Remark.}
The quantity 
\begin{equation} \label{dodici}
C = \int_{-\infty}^{+\infty} \left[ \psi(1-F(y)) - \psi(F(y)) \right] \, f'(y) \, dy 
\end{equation}
is negative or null. It suffices to notice that the function 
$$ g(y) = \psi(1-y) - \psi(y) \qquad 0 \leq y \leq 1$$
is non-decreasing and bounded with
$$ g \left( \frac 12 + \xi \right) = - g \left( \frac 12 - \xi \right) \qquad -\frac 12 \leq \xi \leq \frac 12 $$
and that 
$$ \Cov \left\{ F(Y) \, , \, \frac{f'(Y)}{f(Y)} \right\} = \int_{-\infty}^{+\infty} F(y) \, f'(y) \, dy = - \int_{-\infty}^{+\infty} f^2 < 0. $$
One then gets 
$$ C = \int_{-\infty}^{+\infty} g ( F(y) ) \, f'(y) \, dy = \Cov \left\{ g(F(Y)) \, , \, \frac{f'(Y)}{f(y)} \right\} \leq 0 .$$
When $f$ is an even function, in addition, it immediately follows that
\begin{equation} \label{tredici}
C = -2 \int_{-\infty}^{+\infty} \psi(F(y)) \, (f'(y)) \, dy .
\end{equation}
The function $\psi$ may also happen to be identically null for peculiar sequences of the $x_i$'s, so that it is trivially $C=0.$ This chance may arise when the sequence of the $x_i$'s grows ``too fast" wrt $i,$ for example when $x_i=\alpha^i$ with $\alpha >1.$

\begin{theorem}
If the assumptions of Theorem 2 are met and if $C \neq 0,$ then
$$ \lim_{N \rightarrow + \infty} P \left\{ T (\tilde{\beta} - \beta) \leq b \right\} = \phi \left( - \frac{b}{\sqrt{\frac{1}{24 \, C^2}}} \right) $$
\end{theorem}

\begin{proof}
According to Proposition 1,  $\beta=0$ can be assumed. As the realizations of $G(\underline Y;b)$ are non-increasing and by (\ref{cinque}), 
$$ \left\{ G(\underline Y; \frac bT) < 0 \right\} \quad \Rightarrow \quad \left\{ \tilde{\beta} < \frac bT \right\} \quad
\Rightarrow \quad \left\{ G\left(\underline Y; \frac bT \right) \leq 0 \right\} ,$$
so that, by Theorem 2,
$$ \lim_{N \Rightarrow + \infty} P \left\{ T \tilde{\beta} < b \right\} = \phi \left( -\frac{\sigma_{12}}{\sqrt{2/3}} \right) = \phi \left( - \frac{b}{\sqrt{\frac{1}{24\, C^2}}} \right) . \quad \qed$$ 
\end{proof}

Theorem 3 assures that, under the stated assumptions, the estimator $\tilde{\beta}$ is asymptotically normally distributed with mean $\beta$ and variance
\begin{equation} \label{quattordici}
\Var (\tilde{\beta}) \simeq \frac{1}{24 \, T^2 \, C^2} .
\end{equation}

\section{Asymptotic relative efficiency of $\tilde{\beta}$}
Some comparisons of the proposed estimator $\tilde{\beta}$ with other known estimators will now be conducted in the very important case  $x_i=i,$ $i=1,2,\ldots , N.$ Comparisons with other kinds of sequences can be produced analogously. 

First of all, notice that, in  the considered case, 
$$ \psi(y) = \frac{1}{\sqrt{12}} (2y^3-3y^2) \qquad 0 \leq y \leq 1. $$
To develop suitable comparisons, the asymptotic relative efficiency (ARE) can be used. As known, this technique compares the sample sizes corresponding to two unbiased estimators having  the same asymptotic variance. More specifically, if two estimators $T_1$ and $T_2,$ both asymptotically unbiased for the same parameter $\theta$ and with variances $\Var (T_1)$ and $\Var(T_2),$ need $n_1$ and $n_2$ observations respectively  to obtain the same variance, then
$$ ARE (T_1, T_2) = \lim_{N \rightarrow + \infty} \frac{n_1}{n_2} =  \lim_{N \rightarrow + \infty} \frac{\Var (T_2)}{\Var (T_1)} $$
For the considered sequence of the $x_i$'s, the least squares estimator $\hat{\beta}$ is known to be asymptotically normally distributed with mean $\beta$ and variance
\begin{equation} \label{quindici}
\Var (\hat{\beta}) \simeq \frac{\sigma^2(F)}{T^2}
\end{equation}
where $\sigma^2(F)$ denotes the population variance depending on $F.$ Hence
\begin{equation} \label{sedici}
ARE(\tilde{\beta}, \hat{\beta} ) = 24 \, \sigma^2(F) \, C^2.
\end{equation}
The asymptotic  efficiency of $\tilde{\beta}$ relative to the Theil's estimator $\beta^*$ can be obtained  using Theorem 6.1 in [5] (p. 1385) which, for the considered sequence of the $x_i$'s, states that $\beta^*$ is asymptotically normally distributed with mean $\beta$ and variance 
\begin{equation} \label{diciassette}
\frac{1}{12 \, T^2 \, B^2}
\end{equation}
where $B = \int f^2.$ One then gets
\begin{equation} \label{diciotto}
ARE (\tilde{\beta}, \beta^*) = 2 \, \frac{C^2}{B^2}.
\end{equation}
It is easy to prove that (\ref{sedici}) and (\ref{diciotto}) are invariant under location and scale shifts. The following propositions are of interest:

\begin{proposition}
For any $F$ possessing finite and positive variance, 
$$ ARE (\tilde{\beta} , \hat{\beta} ) > \frac 89 \left( \frac{\sigma(F)}{\Delta(F)} \right)^2 \geq \frac 23 $$
where $\Delta(F)$ denotes the population mean difference depending on $F.$
\end{proposition}

\begin{proof}
After integrating by parts, (\ref{dodici}) gives
$$ C^2 = 12 \, \left( \int_{-\infty}^{+\infty} F(y) \, (1-F(y)) \, f^2(y) \, dy \right)^2 .$$
Now notice that
$$\int_{-\infty}^{+\infty} F \, (1-F) \, f^2 \, dy = \frac{\Delta(F)}{2} \, \int  f^2 \, \frac{2\, F\, (1-F)}{\Delta(F)} .$$
The function
$$ \varphi (y) = \frac{2\, F(y)\, (1-F(y))}{\Delta(F)} \geq 0 $$
is such that, by the definition of $\Delta (F),$
$$ \int_{-\infty}^{+\infty} \varphi(y) \, dy = 1 $$
so that it can be considered as a density function. Hence
$$ \int_{-\infty}^{+\infty} F \, (1-F) \, f^2 \, dy = \frac{\Delta(F)}{2} \, \E \{ f^2(Y) \} $$
where $Y$ is a random variable with density $\varphi(y).$ By a trivial inequality, one has then
$$ \E \{ f^2(Y) \} > \E^2 (f(Y)) = \frac{1}{9 \, \Delta^2(F)} $$
and hence
$$ C^2 > \frac{1}{27 \, \Delta^2(F)} .$$
Formula (\ref{sedici}) gives
$$ ARE(\tilde{\beta}, \beta^*) > \frac 89 \, \left( \frac{\sigma(F)}{\Delta(F)} \right)^2$$
and the proof follows by remembering ([14)) that, for any distribution, 
$$\frac{\sigma(F)}{\Delta(F)} \geq \frac{\sqrt 3}{2} .\quad \qed $$
\end{proof}

\begin{proposition}
For any $F,$
$$ ARE (\tilde{\beta}, \beta^*) < \frac 32 .$$
\end{proposition}

\begin{proof}
One can obtain
\begin{align*}
C^2 &= 12 \left( \int_{-\infty}^{+\infty} (F-F^2) \, f^2 \, dy \right)^2 = \\
&= 12 \left( \frac 14 \, \int f^2 - \int \left( F - \frac  12 \right) ^2 \, f^2  \right)^2 = \\
&= \frac 34 \left( \int f^2 - \int (2F-1)^2 \, f^2 \right)^2 . 
\end{align*}
By using (\ref{diciotto}),
$$ ARE(\tilde{\beta}, \beta^*) = \frac 32 \left( 1- \frac{\int (2F-1)^2 \, f^2}{\int f^2} \right)^2 < \frac 32 . \quad \qed$$
\end{proof}

\vspace{0.5cm}

In the following, the values taken by (\ref{sedici}) and (\ref{diciotto}) will be computed for three specific distributions:
\begin{enumerate}
\item normal
\item double exponential or Laplace
\item Cauchy
\end{enumerate}
which are characterized by a different tail behavior. More specifically, when $|x| \rightarrow + \infty,$ the Cauchy density tends to zero very slowly, in the same manner as $1/x^2;$ the double exponential distribution, instead, has a density tending to zero rather faster than the Cauchy, but more slowly than the normal density. 
 
 \vspace{0.25cm}
\noindent 
1) {\em normal with zero mean and unit variance}

\noindent
By applying (\ref{tredici}), one gets
$$ C = \frac{1}{\sqrt 3} \, \int_0^1 \, (2y^3-3y^2) \, \phi^{-1} (y) \, dy = - \frac{0.3317}{\sqrt 3}$$
being that
\begin{align*}
\int_0^1 y^3 \, \phi^{-1}(y) \, dy &= \int_{-\infty}^{+\infty} z \, \phi^3(z) \, d\phi(z) = 0.2573 \\
\int_0^1 y^2 \, \phi^{-1}(y) \, dy &= \int_{-\infty}^{+\infty} z \, \phi^2(z) \, d\phi(z) = 0.2821 . 
\end{align*}
One has also
$$ B = \int f^2 = \frac{1}{2 \sqrt{\pi}},$$
so that (\ref{sedici}) gives
$$ ARE (\tilde{\beta}, \hat{\beta}) = 8 (0.3317)^2 \simeq 0.88 $$
and (\ref{diciotto}) gives
$$ ARE (\tilde{\beta}, \beta^*) = \frac{8\, \pi}{3} \, (0.3317)^2 \simeq 0.93. $$
Hence, in the normal case the least squares estimator is better than both  $\tilde{\beta}$ and $\beta^*,$ even if none of the latter two estimators shows a substantial loss of efficiency. 

 \vspace{0.25cm}
\noindent 
2) {\em Double exponential}

\noindent
In this case, 
$$ f(y) = \frac 12 \exp \{ - |y| \} \qquad -\infty < y < + \infty, $$
so that $\sigma^2(F)=2.$ After some more computations, one gets
$$ C = - \frac{5}{\sqrt 3 \, 2^4} \qquad \mbox{ and } \qquad B = \frac 14.$$
Hence
$$ ARE (\tilde{\beta}, \hat{\beta}) = \frac{25}{16} \simeq 1.56; \qquad \qquad  ARE (\tilde{\beta}, \beta^*) = \frac{25}{24} \simeq 1.05. $$

 \vspace{0.25cm}
\noindent 
2) {\em Cauchy}

\noindent
Obviously this is an extreme case, because the density
$$ f(y) = \frac{1}{\pi} \, \frac{1}{1-y^2} \qquad -\infty < y < + \infty $$
does not possess finite variance, so  the least squares estimator is not consistent. By definition, one has then
$$ ARE (\tilde{\beta}, \hat{\beta}) = + \infty . $$
However, it makes sense to compare $\tilde{\beta}$ and $\beta^*$. This task results again in favor of $\tilde{\beta}$. Some tedious but trivial computations indeed give 
$$ C = - \frac{\sqrt 3}{2 \, \pi} \left( \frac 13 + \frac{1}{\pi^2} \right); \qquad B = \frac{1}{2 \, \pi} $$
and
$$ ARE (\tilde{\beta}, \beta^*) = 6 \, \left( \frac 13 + \frac{1}{\pi^2} \right)^2 \simeq 1.13.$$

\vspace{0.5cm}
The above results clearly show that  the asymptotic  efficiency of $\tilde{\beta}$ relative to $\hat{\beta},$ but also to $\beta^*,$ tend to grow as distributions with more and more heavy tails are considered.  

\section*{Appendix}
To prove Theorem \ref{t2} of section \ref{s7}, some preliminary results will be considered. Let $f$ be a probability density function with support in $\Re$ and define the two probability measures
$$ P_N(A) = \int_A \, \prod_{1 \leq i \leq N} f(y_i) \quad \mbox{ and } \quad
Q_N(A) = \int_A \, \prod_{1 \leq i \leq N} f \left( y_i + \frac{b}{T} (x_i-\bar x) \right) $$
where $x_1 < x_2 < \ldots < x_N$ as usual, $b \neq 0$ is finite, $A$ is any event and
$$ T^2 = \sum_{1 \leq i \leq N} (x_i - \bar x)^2 \qquad M = \max_{1 \leq i \leq N} (x_i - \bar x)^2 .$$

\begin{lemma} \label{lemma1app}
{\rm ([15], p. 208) ([16, p. 1134])}\\
If the vector
$$ \left( \sqrt N \, G(\underline Y;0) \: , \: \log \frac{\displaystyle{\prod_{1 \leq i \leq N}} \, f \left( Y_i + \frac{b}{T} (x_i - \bar x) \right) }{\displaystyle{\prod_{1 \leq i \leq N}} f(Y_i)} \right) $$ 
converges, with the measure $P_N$, to the normal distribution with parameters
$$ \left( \mu \, , \,  - \frac 12 \sigma_2^2 \, , \, \sigma_1^2 \, , \, \sigma_2^2 \, , \, \sigma_{12} \right) $$
then the variable 
$$ \sqrt N \, G(\underline Y;0)$$
converges, with the measure $Q_N$, to the normal distribution with mean $\mu + \sigma_{12}$ and variance $\sigma_1^2.$
\end{lemma}

\begin{lemma} \label{lemma2app}
{\rm ([15], p. 213), ([16], p. 1136).}\\
If 
$ \displaystyle{I(f) = \int \left( \frac{f'}{f} \right)^2  f  < \infty}$ and if 
$\displaystyle{\frac{T^2}{M} \rightarrow \infty },$
then
\begin{eqnarray*}
\lefteqn{
{P_N\lim}_{N \rightarrow \infty} \Bigg( \log \frac{\displaystyle{\prod_{1 \leq i \leq N}} f \left( 
Y_i + \frac{b}{T} (x_i-\bar x) \right)}{\displaystyle{\prod_{1 \leq i \leq N}} \, f(Y_i)} - \frac bT \sum_{1 \leq i \leq N}
(x_i - \bar x) \frac{f'(Y_i)}{f(Y_i)} } && \\
&& \hspace{8.5cm} + \frac b2 \, I(f) \Bigg) =0 
\end{eqnarray*}
where $P_N \lim$ denotes the limit in $P_N$--probability.
\end{lemma}

\noindent
{\em Remark to Lemma \ref{lemma2app}}
\\
According to the measure $P_N,$ the variable 
$$ \frac bT \sum_{1 \leq i \leq N} (x_i - \bar x) \, \frac{f'(Y_i)}{f(Y_i)} $$
has the following variance
$$ \mbox{Var} \left(  \frac bT \sum_{1 \leq i \leq N} (x_i - \bar x) \, \frac{f'(Y_i)}{f(Y_i)}  \right) = b^2 \, I(f) < \infty $$
and expectation ([16], p. 1125)
$$ \mbox{E} \left(  \frac bT \sum_{1 \leq i \leq N} (x_i - \bar x) \, \frac{f'(Y_i)}{f(Y_i)}  \right) = 0 .$$
Moreover, the variable
$$ \sum_{1 \leq i \leq N} Z_{2i} = \sum_{1 \leq i \leq N} \frac{\displaystyle{\frac bT  \, (x_i - \bar x) \,  \frac{f'(Y_i)}{f(Y_i)}}}{\sqrt{b^2 I (f)}} $$
satisfies the Lindeberg-Feller condition. Indeed, after defining
$$ Z_{2i} (\delta) = Z_{2i} \, \, s \left( |Z_{2i}| - \delta \right) $$
where $\delta >0$ and $s(x)$ equals 1 if $x\geq 0$ and $0$ elsewhere, such a condition can be written as
$$ \sum_{1 \leq i \leq N} \mbox{E} \left( Z_{2i}^2 (\delta) \right) =  \sum_{1 \leq i \leq N} \int_{|z| \geq \delta}
z^2 \, d P_N \left\{ \frac bT  (x_i - \bar x) \frac{f'(Y_i)}{f(Y_i} \leq z \, \sqrt{b^2 \, I(f)} \right\} \rightarrow 0 .$$
However, by putting $t = z \, \sqrt{b^2 \, I(f)},$ one gets
\begin{eqnarray*}
\lefteqn{\sum_{1 \leq i \leq N} \frac{1}{b^2 \, I(f)} \, \int_{|t| \geq \delta \sqrt{b^2 \, I(f)}} \, t^2 \, d P_N \left\{ 
\frac{f'(Y_i)}{f(Y_i}  \leq t \right\} = } && \\
&& = \frac{1}{T^2} \,  \sum_{1 \leq i \leq N} \frac{(x_i- \bar x)^2}{I(f)} \int_{|y| \geq \delta \sqrt{I(f)} \left| \frac{T}{x_i - \bar x} \right|} \, y^2 \, d P_N \left\{ 
\frac{f'(Y_i)}{f(Y_i}  \leq y \right\} \rightarrow  0
\end{eqnarray*}
because, by hypothesis, 
$$ \frac{T^2}{M} \rightarrow + \infty \quad \Rightarrow \quad \left| \frac{T}{x_i- \bar x} \right| \Rightarrow + \infty.$$

\begin{lemma} \label{lemma3app}
If $F'=f$ and if 
$$ \hat G(\underline Y;0) = \frac{2N}{D} \sum_{1 \leq i \leq N} \left( \left| 1 - \frac iN - F(Y_i) \right| - \left| \frac iN - F(Y_i) \right| \right) $$
then 
$$ {P_N \lim}_{N \rightarrow \infty} \, \sqrt N \, (G(\underline Y;0) - \hat G(\underline Y;0)) = 0. $$
\end{lemma}

\begin{proof}
 By using the identity
 $ |x| = 2x \, s(x) - x $ $(x \in \Re),$ the definition in (\ref{due}) and the expression of $\hat G(\underline Y; 0),$ one gets 
 $$ \sqrt N (G(\underline Y;0) - \hat G (\underline Y;0)) = A_N + B_N + C_N + D_N $$
 where
 \begin{align*}
 A_N & =\frac{4 \, N^{3/2}}{D} \, \sum_{1 \leq i \leq N} \left( 1- \frac iN - F(Y_i) \right) \left[ s(N+1-i-R(Y_i)) \right. \\
 & \hspace{8cm} \left. - s(N-i-N\, F(Y_i)) \right] \\
 B_N & =- \frac{4 \, N^{3/2}}{D} \, \sum_{1 \leq i \leq N} \left(\frac iN - F(Y_i) \right)
 \left[ s(i-R(Y_i)) - s(i-N\, F(Y_i)) \right] \\
 C_N & =\frac{4 \, N^{3/2}}{D} \, \sum_{1 \leq i \leq N} \left( \frac{R(Y_i)}{N} - F(Y_i) \right) \left[ s(i - R(Y_i)) - s(N+1-i-R(Y_i)) \right] \\
 D_N & =\frac{2 \, N^{1/2}}{D} \, \sum_{1 \leq i \leq N}  s(N+1-i-R(Y_i)) .
 \end{align*}
 It follows that
 \begin{align*}
 \mbox{E} \{ |D_N| \} & \leq \frac{2 \, N^{3/2}}{D} \, \rightarrow 0 \\
 \mbox{E} \{ |B_N| \} & \leq \frac{4 \, N^{3/2}}{D} \sum_{1 \leq i \leq N} \mbox{E} \left\{
 \left| \frac iN - F(Y_i) \right| \, \, \left| s(i-R(Y_i)) - s(i-NF(Y_i)) \right| \right\}
 \end{align*}
 By using the joint distribution of $(Y_i, R(Y_i)),$ that is
 \begin{eqnarray*}
 \lefteqn{ \mbox{Pr} \{ R(Y_i) = r \, ; \, y < Y_i < y+dy \} = } &&\\
 && = \frac 1N \, g_{Y_{(r)}} (y) \, dy = \\
 && 	= \frac 1N \frac{N!}{(N-r)! \, (r-1)!} \,  [F(y)]^{r-1} \, [1-F(y)]^{N-r} \, f(y) \, dy 
 \quad r=1,2, \ldots, N; y \in \Re , 
 \end{eqnarray*}
 one gets
 \begin{eqnarray*}
 \lefteqn{
 \mbox{E} \left\{
 \left| \frac iN - F(Y_i) \right| \, \, \left| s(i-R(Y_i) - s(i-NF(Y_i)) \right| \right\} = } &&\\
 && = \frac 1N \sum_{1 \leq r \leq N} \int_{- \infty}^{+ \infty} 
 \left| \frac iN - F(y) \right| \, \, \left| s(i-r) - s(i-NF(y)) \right| \, g_{Y_{(r)}} (y) \, dy = \\
 && = \sum_{1 \leq r \leq N} \int_0^1 
 \left| \frac iN - v \right| \, \, \left| s(i-r) - s(i-Nv) \right|
 {N-1 \choose r-1} \, v^{r-1} \, (1-v)^{N-r} \, dv = \\
  && =\int_0^1 \left| \frac iN - v \right| \left[ \left( 1- s(i-Nv) \right) \, \Pr \{ 
 U_{(i)} > v \} + s(1-Nv) \, \Pr \{ U_{(i)} \leq r \} \right] \, dv 
 \end{eqnarray*}  
 where $U_{(i)}$ is the $i-$th order statistic of a $(N-1)-$sized random sample drawn from a uniform population in $(0,1).$ By partitioning the integration interval, after some trivial passages, one gets
 \begin{eqnarray*}
 \lefteqn{
 \mbox{E} \left\{
 \left| \frac iN - F(Y_i) \right| \, \, \left| s(i-R(Y_i) - s(i-NF(Y_i)) \right| \right\} = } &&\\
 && = \frac{i^2}{2 \, N^2} - \frac iN \, \int_0^1 \Pr \{ U_{(i)} >v \} \, dv + \int_0^1 v \, \, \Pr \{ 
 U_{(i)} >v \} \, dv = \\
 && = \frac{i^2}{2 \, N^2} - \frac i N \, \mbox{E} \left( U_{(i)} \right) + \frac 12 \, \mbox{E}
 \left( U_{(i)}^2 \right) = \\
 && = \frac{i \, (N-i)}{2 \, N^2 (N+1)} \qquad \qquad 1 \leq i \leq N 
 \end{eqnarray*}  
so that
$$ \mbox{E} \{ |B_N| \} \leq \frac{2 \, N^{3/2}}{D} \, \sum_{1 \leq i \leq N} \frac{i(N-i)}{N^2(N+1)} \quad\rightarrow 0 .$$
 By following similar steps, one can prove that
 $$ \mbox{E} \{ |A_N| \} \rightarrow 0 . $$
 Now let 
 $$ S_{i,N} = \left( \frac{R(Y_i)}{N} - F(Y_i) \right) \left[ s ( i - R(Y_i) ) - s(N+1-i-R(Y_i)) \right] \qquad i =1, \ldots , N .$$
 and simply consider that
 $$ \mbox{E} (C_N) = \frac{4 \, N^{3/2}}{D} \, \sum_{1 \leq i \leq N} \mbox{E} \left( S_{i,N} \right)=0 $$
 Moreover, 
 \begin{equation} \label{asterisco}
 \mbox{Var} (C_N) = \frac{16 \, N^3}{D^2} \sum_{1 \leq i \leq N} \mbox{E} \left( S_{i,N}^2 \right) + \frac{16 \, N^3}{D^2} \sum_{i \neq j} \, \mbox{E} \left( S_{i,N} \, S_{j,N} \right) 
 \end{equation}
 and
 \begin{eqnarray*}
 \sum_{1 \leq i \leq N} \mbox{E} (S_{i,N}^2) &=& \sum_{1 \leq i \leq N} \frac 1N \sum_{1 \leq r \leq N} 
 \int_{-\infty}^{+ \infty} \left( \frac rN - F(y) \right)^2   \cdot \\
 && \hspace{3cm} \cdot \, \left[ s(i-r) -s(N+1-i-r) \right]^2 \, g_{Y_{(r)}} (y) \, dy \\
 & \leq & \sum_{1 \leq r \leq N} \int_{-\infty}^{+\infty} \left( \frac rN - F(y) \right)^2 \, 
 g_{Y_{(r)}}(y)\,dy \\  
 & \leq & \sum_{1 \leq r \leq N} \mbox E \left\{ \left( F(Y_{(r)}) -\frac rN \right)^2 \right\} \, \rightarrow \, 
 A < +\infty
 \end{eqnarray*}
 being that
 $$ \mbox E \left\{ \left( F(Y_{(r)}) - \frac rN  \right)^2 \right\} = \frac{r(N-r+1)}{(N+2)\, (N+1)^2} + 
 \frac{r^2}{N^2 \, (N+1)^2} \qquad \forall r = 1,2, \ldots, N $$
 The first summand in the rhs of (\ref{asterisco}) thus tends to zero. Moreover,
 \begin{eqnarray*}
 \lefteqn{ \left| \sum_{i \neq j} \mbox E (S_{i,N} \, S_{j,N}) \right|  = } && \\
 && = \left| \frac{1}{N(N-1)} \, \sum_{i \neq j} \, \sum_{r \neq k} \int_{-\infty}^{+\infty}
 \int_{-\infty}^{+\infty} \left( \frac rN - F(x) \right) \, \left( \frac kN - F(y) \right) \right. \\
 && \hspace{2cm} 
 \left( s(i-r) - s(N+1-i-r) \right) \, \left( s(j-k) - s (N+1-j-k) \right) \\
 && \hspace{7.5cm} g_{Y_{(r)},Y_{(k)}}
(x,y) \, dx \, dy \Big| = \\  
&& = \left| \frac{1}{N \, (N-1)} \sum_{r \neq k} \left[ \int_{-\infty}^{+\infty} \int_{-\infty}^{+\infty}
\left( \frac rN - F(x) \right) \, \left( \frac kN - F(y) \right) \, g_{Y_{(r)}, Y_{(k)}} (x,y) \right. \right. \\
&& \hspace{0.5cm} dx \, dy \Big] \sum_{i \leq i \leq N} \left( s(i-r) -s(N+1-i-r) \right) \, \left( s(i-k) - s(N+1-i-k) \right) \big| .
 \end{eqnarray*}
 Now, as
 \begin{eqnarray*}
 \lefteqn{ \int_{-\infty}^{+\infty} \int_{-\infty}^{+\infty}
\left( \frac rN - F(x) \right) \, \left( \frac kN - F(y) \right) \, g_{Y_{(r)}, Y_{(k)}} (x,y) \, dx \, dy = } \\
&& = \mbox{Cov} \left\{ F(Y_{(r)}) \, , \, F(Y_{(k)}) \right\} + \frac{rk}{N^2 \, (N+1)^2} = \\
&& = \left\{ 
\begin{array}{ll}
\dfrac{r(N+1-k)}{(N+2)\, (N-1)^2} + \dfrac{rk}{N^2 \, (N+1)^2} & \hspace{1.5cm} r<k \\
\dfrac{k(N+1-r)}{(N+2)\, (N+1)^2} + \dfrac{rk}{N^2 \, (N+1)^2} & \hspace{1.5cm} r>k
\end{array} \right. 
\end{eqnarray*}
one obtains
 \begin{eqnarray*}
 \lefteqn{ \left| \sum_{i \neq j} \mbox{E} (S_{i,N} \, S_{j,N}) \right| \leq} \\
 && \leq \frac{1}{N-1} \sum_{r \neq k} \left\{ \mbox{Cov} (F(Y_{(r)}) \, , \, F(Y_{(k)}) +
 \frac{rk}{N^2 \, (N+1)^2} \right\} \, \rightarrow \, B < + \infty 
 \end{eqnarray*}
 so that  the second summand in (\ref{asterisco}) tends to zero as well. The proof follows then by applying Thchebycheff's inequality in its suitable form to the four variables. $\qed$
\end{proof}
 
\noindent
{\em Remark to Lemma \ref{lemma3app}}
\\
Lemma \ref{lemma3app} makes it possible to obtain the asymptotic distribution of Gini's cograduation index under indifference in an alternative way with respect to a former paper ([17]). Indeed, Lemma \ref{lemma3app} assures that $\sqrt N \,  G(\underline Y;0)$ is asymptotically equally distributed as $\sqrt N \, \hat G (\underline Y; 0),$ for which the classical limit theorems can be applied, because it can be regarded as a sum of independent variables. As a matter of fact, the variable
$$ \sqrt N \, \hat G(\underline Y; 0) = \frac{2 \, N^{3/2}}{D} \, \sum_{1 \leq i \leq N}
\left( \left| 1-\frac iN - F(Y_i) \right| - \left| \frac iN - F(Y_i) \right| \right) $$ 
has the following mean and variance
$$ \mbox{E} (\sqrt N \, \hat G(\underline Y;0)) = 0 \qquad 
\mbox{Var} (\sqrt N \, \hat G(\underline Y;0)) \simeq \frac 23 $$
and, by letting
$$ \sum_{1 \leq i \leq N} \dfrac{\frac{2 \, N^{3/2}}{D} \, \left( | 1 - \frac iN - F(Y_i) |
- | \frac iN - F(Y_i) | \right)}{\sqrt{\frac 23}} = \sum_{1 \leq i \leq N} Z_{1i} $$
and, for every $\delta >0,$
$$ Z_{1i} (\delta) = Z_{1i} \: s(|Z_{1i} - \delta|),$$
the Lindeberg condition is satisfied:
$$ \sum_{1 \leq i \leq N} \mbox{E} (Z_{1i}^2 (\delta)) \, \rightarrow \, 0 \qquad \forall \delta >0 $$

\begin{lemma} \label{lemma4app}
If $F' = f,$ $I(f) < +\infty,$ $\int |f'| < + \infty$ and $\frac{T^2}{M} \rightarrow + \infty,$
then the vector
$$ \left( \sqrt N \, G (\underline Y;0) \, , \, \log \frac{\displaystyle{\prod_{1 \leq i \leq N}} \, f 
\left( Y_i + \frac bT (x_i - \bar x) \right) }{\displaystyle{\prod_{1 \leq i \leq N}} f(Y_i)} \right) $$
converges in distribution, with the measure $P_N,$ to the bivariate normal with parameters
$$ \left( 0, \quad -\frac{b^2}{2} I(f), \quad \frac 23 , \quad b^2 I(f), \quad \sigma_{12} \right) $$
where
\begin{align*}
\sigma_{12} & = 4b \, \int_{-\infty}^{+\infty} \left[ \psi(1-F(y)) - 
\psi(F(y)) \right] \, f'(y) \, dy \\
\psi(y) & = \lim_{N \rightarrow + \infty} \, \frac{1}{N^{3/2} \, T} \, \sum_{i=1}^{[Ny]} (x_i - \bar x) (Ny-i) \qquad \qquad 0 <y \leq 1 
\end{align*}
\end{lemma}

\begin{proof}
By following lemmas \ref{lemma2app} and \ref{lemma3app}, it suffices to show that the vector 
$$ \left( \sqrt{N} \, \hat G (\underline Y;0) \, , \, \frac bT \sum_{1 \leq i \leq N} (x_i -\bar x) \dfrac{f'(Y_i)}{f(Y_i)} \right) $$
converges in distribution to the bivariate normal with parameters
$$ \left( 0, \quad 0, \quad \frac 23 , \quad b^2 I(f), \quad \sigma_{12} \right) .$$
By the remarks following Lemma \ref{lemma2app} and Lemma \ref{lemma3app}, the limiting distribution  surely takes the first  four parameters listed above. Moreover, consider that
\begin{eqnarray*}
\lefteqn{\mbox{Cov} \left\{ \sqrt N \hat G (\underline Y;0)\, , \, \frac bT \sum_{1 \leq i \leq N} (x_i - \bar x) \dfrac{f'(Y_i)}{f(Y_i)} \right\} = } \\
&& = \frac{2 \, N^{3/2}}{D \, T} \, b \, \int_{-\infty}^{+\infty} f'(y) \, 
\left[ \sum_{1 \leq i \leq N} (x_i - \bar x) \left( \left| 1 - \frac iN - F(y) \right| - \left| \frac iN - F(y) \right| \right) \right] \, dy = \\
&& = 4b \, \int_{- \infty}^{+ \infty} f'(y) \left[ \frac{N^{3/2}}{D \, T} \left( \sum_{i=1}^{[N\, (1-F(y))]} 
(x_i - \bar x) \left( 1-F(y) - \frac iN \right) + \right. \right. \\
&& \hspace{6cm} 
- \left. \left. \sum_{i=1}^{[N\, F(y)]} (x_i - \bar x) \left( F(y) - \frac iN \right) \right) \right] \, dy 
\end{eqnarray*}
By passing to the limit (with $N$) under the integral sign, one gets
\begin{align*}
\sigma_{12} & = \lim_{N \rightarrow + \infty} \mbox{Cov} 
\left\{ \sqrt N \hat G (\underline Y;0)\, , \, \frac bT \sum_{1 \leq i \leq N} (x_i - \bar x) \dfrac{f'(Y_i)}{f(Y_i)} \right\} = \\
& = 4b \, \int_{- \infty}^{+ \infty} \left[ \psi(1-F(y)) - \psi(F(y)) \right] \, f'(y) \, dy = \\
& = 4b \, \int_0^1 \left[ \psi(1-v) - \psi(v) \right] \, \dfrac{f'(F^{-1}(v))}{f(F^{-1}(v))} \, dv
\end{align*}
To prove that the limiting distribution is normal, one can then show that, for every real $\lambda_1$ and $\lambda_2,$ the following variable is asymptotically normally distributed:
$$  \lambda_1 \, \sqrt{N} \, \hat G (\underline Y;0)+ \lambda_2 \,  \frac bT \sum_{1 \leq i \leq N} (x_i -\bar x) \dfrac{f'(Y_i)}{f(Y_i)} .$$
However, as both the variables
$$ \lambda_1 \, \sqrt{N} \, \hat G (\underline Y;0) \quad \mbox{ and } \quad 
  \lambda_2 \,  \frac bT \sum_{1 \leq i \leq N} (x_i -\bar x) \dfrac{f'(Y_i)}{f(Y_i)} $$
satisfy the Lindeberg condition, one can get the aimed result as in [15], page 218. $\qed$
\end{proof}

\vspace{0.5cm}

The proof of Theorem \ref{t2} in section \ref{s7}  now immediately follows from lemmas \ref{lemma1app} and \ref{lemma4app}. Indeed, for every real $z,$
\begin{align*}
Q_N \left\{ \sqrt N \, G(\underline Y;0) \leq z \right\} & = \int_{\left\{ \sqrt N \, G(\underline y ; 0) \leq z \right\}} \prod_{1 \leq i \leq N} \, f \left( y_i + \frac bT (x_i - \bar x) \right) = \\
& = \int_{\left\{ \sqrt N \, G \left(\underline y ; \frac bT \right) \leq z \right\}} \prod_{1 \leq i \leq N} \, 
f(y_i) = \\
& = P_N \left\{ \sqrt N \, G \left(\underline Y ; \frac bT \right) \leq z \right\} 
\end{align*}
and, by lemmas \ref{lemma1app} and \ref{lemma4app},
$$ \lim_{N \rightarrow + \infty} Q_N \left\{ \sqrt N \, G(\underline Y;0) \leq z \right\} = \phi \left( \frac{z-\sigma_{12}}{\sqrt{2/3}} \right) .$$

\newpage
\section*{References}

\noindent
[1] \textsc{F. Eicker}, \textit{Asymptotic normality and consistency of least squares estimators for families of linear regressions.} Ann. of Math. Stat., \textit{34}, 1963, pp. 447--456.
\vspace{0.25cm}

\noindent
[2] \textsc{A. M. Mood}, \textit{An introduction to the theory of statistics.} McGraw-Hill Book Co., New York,  1950.
\vspace{0.25cm}

\noindent
[3] \textsc{J. N. Adichie}, \textit{Estimators of regression parameters based on rank tests.} Ann. of Math. Stat., \textit{38}, 1967, pp. 894--904.
\vspace{0.25cm}

\noindent
[4] \textsc{H. Theil}, \textit{A rank-invariant method of linear and polynomial regression analysis, I, II, III.} Nederl. Akad. Westensch. Proc., \textit{53}, 1950, pp. 386--392, 521--525, 1397--1412.
\vspace{0.25cm}

\noindent
[5] \textsc{P. K. Sen}, \textit{Estimates of regression coefficient based on Kendall's tau.} Journ. Amer. Stat. Ass., \textit{63}, 1968, pp. 1379--1389.
\vspace{0.25cm}

\noindent
[6] \textsc{R. J. Beran}, \textit{On distribution-free statistical inference with upper and lower probabilities.} Ann. of Math. Stat., \textit{42}, 1971, pp. 157--158.
\vspace{0.25cm}

\noindent
[7] \textsc{H. L. Koul}, \textit{Asymptotic behavior of Wilcoxon type confidence regions in multiple linear regression.} Ann. of Math. Stat., \textit{40}, 1969, pp. 1950--1979.
\vspace{0.25cm}

\noindent
[8] \textsc{R. V. Hogg} and \textsc{R. H. Randles}, \textit{Adaptive distribution free regression methods and their applications.} Technometrics, \textit{17}, 1975, pp. 399--407.
\vspace{0.25cm}

\noindent
[9] \textsc{P. Muliere}, \textit{Una nota intorno al coefficiente di correlazione tra l'indice G di cograduazione di Gini e l'indice} $\tau$ \textit{di Kendall.} Giornale degli Economisti e Annali di Economia, 1976, pp. 627--633.
\vspace{0.25cm}

\noindent
[10] \textsc{A. Herzel}, \textit{Sulla distribuzione campionaria dell'indice di cograduazione del Gini.} Metron, \textit{30}, 1972, pp. 137--153.
\vspace{0.25cm}

\noindent
[11] \textsc{C. Gini}, \textit{Sul criterio di concordanza tra due caratteri.} Atti del Reale Istituto Veneto di Scienze, Lettere ed Arti, 1915-16, pp. 309--331.
\vspace{0.25cm}

\noindent
[12] \textsc{W. Hoeffding}, \textit{A non parametric test of independence.} Ann. of Math. Stat., \textit{19}, 1948, pp. 546--557.
\vspace{0.25cm}

\noindent
[13] \textsc{J. L. Hodges} and \textsc{E. L. Lehmann}, \textit{Estimates of location based on rank tests.} Ann. of Math. Stat., \textit{34}, 1963, pp. 598--611.
\vspace{0.25cm}

\noindent
[14] \textsc{B. Michetti} and \textsc{G. Dall'Aglio}, \textit{La differenza semplice media.} Statistica, \textit{17}, 1957, pp. 159--255.
\vspace{0.25cm}

\noindent
[15] \textsc{J. H\'{a}jek} and \textsc{Z. \v{S}id\'{a}k}, \textit{Theory of rank tests.} Academic Press, London,  1962.
\vspace{0.25cm}

\noindent
[16] \textsc{J. H\'{a}jek}, \textit{Asymptotically most powerful rank-order tests.} Ann. of Math. Stat., \textit{33}, 1962, pp. 1124--1147.
\vspace{0.25cm}

\noindent
[17] \textsc{D. M. Cifarelli} and \textsc{E. Regazzini}, \textit{On a distribution-free test of independence based on Gini's rank association coefficient.} Recent developments in Statistics. Proceedings of the European Meeting of Statisticians (Grenoble, 6-11 sept. 1976), North-Holland, Amsterdam, 1977, pp. 375--385.
\vspace{0.25cm}

\end{document}